\documentclass[letterpaper, 10 pt, conference]{ieeeconf}  % Comment this line out

\IEEEoverridecommandlockouts                              % This command is only
\usepackage{cite}

\usepackage{amsmath,amssymb,amsfonts}
\usepackage{amsthm}
\usepackage{algorithm2e}
\usepackage{graphicx}
\usepackage{makecell}
\usepackage{caption}
\usepackage{textcomp}
\usepackage{comment}

\DeclareMathOperator*{\argmax}{arg\,max}

\usepackage{xcolor}
\def\BibTeX{{\rm B\kern-.05em{\sc i\kern-.025em b}\kern-.08em
    T\kern-.1667em\lower.7ex\hbox{E}\kern-.125emX}}                                       
\newcommand{\rom}[1]{\uppercase\expandafter{\romannumeral #1\relax}}
    
\theoremstyle{plain}
\newtheorem{theorem}{Theorem}

\newtheorem{lemma}{Lemma}

\newtheorem{definition}{Definition}
\newtheorem{cat}{Category}
\newtheorem{ass}{Assumption}

\theoremstyle{remark}
\newtheorem{remark}{Remark}
\newtheorem*{galg}{Greedy algorithm}

\usepackage{graphics} % for pdf, bitmapped graphics files

\title{\LARGE \bf
An Improved Greedy Curvature Bound in Finite-Horizon String Optimization with Application to a Sensor Coverage Problem 
}
\author{Brandon Van Over, Bowen Li, Edwin K.P. Chong, and Ali Pezeshki% <-this % stops a space
\thanks{This work is supported in part by the AFOSR under award FA8750-20-
2-0504 and by the NSF under award CNS-2229469.}% <-this % stops a space
\thanks{B. Van Over, B. Li, E. K. P. Chong and A. Pezeshki are with the Department of Electrical and Computer Engineering, Colorado State University, Fort Collins, CO 80523, USA. Email:
        {\tt\small \{b.van\_over, bowen.li, edwin.chong, ali.pezeshki\}@colostate.edu}}%
}

\begin{document}

\maketitle
\thispagestyle{empty}
\pagestyle{empty}

%%%%%%%%%%%%%%%%%%%%%%%%%%%%%%%%%%%%%%%%%%%%%%%%%%%%%%%%%%%%%%%%%%%%%%%%%%%%%%%%
\begin{abstract}
We study the optimization problem of choosing strings of finite length to maximize string submodular functions on string matroids, which is a broader class of problems than maximizing set submodular functions on set matroids. We provide a lower bound for the performance of the greedy algorithm in our problem, and then prove that our bound is superior to the greedy curvature bound of Conforti and Cornuéjols. Our bound has lower computational complexity than most previously proposed curvature bounds. Finally, we demonstrate the strength of our result on a sensor coverage problem. 
\end{abstract}

%%%%%%%%%%%%%%%%%%%%%%%%%%%%%%%%%%%%%%%%%%%%%%%%%%%%%%%%%%%%%%%%%%%%%%%%%%%%%%%%
\section{INTRODUCTION}
In sequential decision-making and optimal control, we
commonly face the problem of choosing a set of actions
over a finite horizon to maximize a given objective function. In numerous cases those objective functions display the diminishing return property, otherwise known as submodularity, in many real-life applications such as document summarization \cite{lin2011class}, feature selection\cite{krause2012near}, and optimizing viral marketing strategies on social media \cite{kempe2003maximizing}. However, computing the optimal solution of this class of problems becomes more intractable with the increasing size of the action space and growing horizon. The aforementioned unmanageable computations have motivated people to approximate optimal solutions. One widely used approximation method is the greedy algorithm, in which we select the action maximizing the incremental gain of the objective function at each step.

Previous work has mostly been devoted to providing performance guarantees for the greedy algorithm applied to submodular set functions on set matroids such as in \cite{calinescu2011maximizing}, as well as the seminal result from Nemhauser \textit{et al.}, which provides a guarantee of the greedy algorithm yielding an objective value that is at least $1/2$ of the optimal one over a finite general set matroid, and $1-e^{-1}$ over a finite uniform set matroid \cite{nemhauser1978analysis}. Later on, different types of computable \textit{curvatures} were introduced by Conforti and Cornuéjols in \cite{conforti1984submodular}. Of particular relevance was that of the greedy curvature, which in numerous cases allowed for improvement upon the $1-e^{-1}$  lower bound in \cite{nemhauser1978analysis}. More recently, other notions of curvature have been developed such as elemental curvature \cite{wang2016approximation}, partial curvature \cite{liu2019improved}, and extended greedy curvature \cite{welikala2022new}, which are computed and their subsequent values are used in a formula bounding the performance of the greedy algorithm. 

Others have considered more general problems of maximizing functions on strings such as \cite{streeter2008online} and \cite{zhang2015string}. The similarity between the present work and those of \cite{zhang2015string} and \cite{streeter2008online} is worth noting in that both works investigate the problem of maximizing functions on strings with some conditions similar to string submodularity. In fact, our notions of string matroid and string submodularity coincide exactly with \cite{zhang2015string}. The main distinction of our results  is that the bound we present is computable, while those of \cite{streeter2008online} and \cite{zhang2015string} rely on curvatures that are intractable to compute for large domains.

In our work, we improve upon the greedy curvature bound presented by Conforti and Cornuéjols in \cite{conforti1984submodular}, and show that the resulting bound can be reduced to a simple quotient. The strength of our work lies in the following three attributes:
\begin{itemize}
    \item Our bound has a wide scope of applicability to string submodular functions on two classes of finite rank string matroids.
    \item Our bound is always computable.  
    \item Our bound is provably superior to the greedy curvature bound in \cite{conforti1984submodular}.
\end{itemize}

The rest of the paper is organized as follows. Section \rom{2} introduces all the mathematical preliminaries regarding the string optimization problems and how to translate between set and string matroids. Section \rom{3} presents some previous results on the performance bound of greedy algorithm, including the classical bound and some curvature bounds. Our bound and the relevant theoretical contributions are shown in Section \rom{4}, followed by the numerical demonstration on sensor coverage problem in Section \rom{5}. 

\section{Mathematical Preliminaries}
\subsection{Problem Setup}

We define a \textbf{string} of length $n$ comprised of elements from a set $\mathcal{A}$ to be an ordered $n$-tuple $A = (a_1, ..., a_n)$ where $a_i \in \mathcal{A}$ for all $1 \leq i \leq n$, and the length is denoted by $|A| = n$. The empty string, denoted by $\emptyset$, will have length zero.  In our work the set $\mathcal{A}$ will be referred to as the \textbf{action set} and any element $a \in \mathcal{A}$ will be referred to as an \textbf{action}. We then let $\mathcal{A}^*$ be the \textbf{universal action set}, i.e. the set of all strings comprised of actions from $\mathcal{A}$ of arbitrary length. When we want to restrict to the subset of all strings whose length is less than or equal to some fixed finite number $K$, the set is denoted by $\mathcal{A}^*_K$. 

On strings we define the binary operation of \textit{concatenation}, denoted by $\cdot$, which takes strings $A= (a_1, ..., a_k)$ and $B = (b_1, ..., b_l)$ belonging to $\mathcal{A}^*$ and produces the new string \[
A \cdot B = (a_1, ..., a_k, b_1, ..., b_l) \in \mathcal{A}^*.
\]
We say that two strings $A, B \in \mathcal{A}^*$ are equal if $|A| = |B|$ and $a_i = b_i$
 for all $i=1, \ldots, |A|$. We define the partial ordering $\preccurlyeq$ on $\mathcal{A}^*$ to be $A \preccurlyeq B$ if and only if $B = A\cdot L$ for some $L \in \mathcal{A}^*$.
 
As in \cite{zhang2015string} we say that a function $f: \mathcal{A}^* \to \mathbb{R}$ is \textbf{string submodular} if 

\begin{enumerate}
    \item $f$ has the \textit{forward monotone} property, i.e., $\forall A \preccurlyeq B \in \mathcal{A}^*, f(A) \leq f(B)$.
    \item $f$ has the \textit{diminishing returns} property, i.e., $\forall A \preccurlyeq B \in \mathcal{A}^*, \forall a \in \mathcal{A} \text{ that are feasible at } A \text{ and } B, f(A \cdot a)-f(A) \geq f(B \cdot a) - f(B)$. 
\end{enumerate}

\begin{remark}
For brevity of notation, we denote $f(A \cdot a) - f(A)$ by $\rho_{a}(A)$ and call this quantity the \textbf{discrete derivative} on string $A$. We also say that when $A \cdot a $ belongs to the restricted domain of our function, then $a$ is \textbf{feasible} at $A$. %For any subset $Y \subset \mathcal{A}^*$, we say that $f: Y \to \mathbf{R}$ is string submodular on $Y$ if all of the above properties hold for $f$ on $Y$. 
Besides, we will assume that $f(\emptyset)=0$ for the rest of the paper, since if not, we may replace $f$ by $f-f(\emptyset)$.  
\end{remark}

In the present work we consider the restriction of string submodular functions $f: \mathcal{A}^* \to \mathbb{R}$ to subsets $\mathcal{A}^*_K \subset \mathcal{A}^*$ for a finite $K$ and seek to solve the optimization problem:

\[ \max_{A \in \mathcal{A}^*_K}f(A). 
\]
The solution to the above problem will be referred to as the \textit{optimal strategy} and is denoted $O_K = (o_1, ..., o_K)$.  Unfortunately finding the optimal strategy is intractable in general, and so we approximate $f(O_K)$ using the output of the greedy algorithm. The greedy algorithm is defined as follows:

\begin{galg}
A string $G_K = (g_1, ..., g_K)$ is called \textit{greedy} if $\forall i = 1, ..., K$,

\[
g_i \in \text{arg}\max_{g\in \mathcal{A}}f((g_1, ..., g_{i-1}, g)).
\]
\end{galg}

To get an idea of how close the greedy strategy is to optimal, we seek to place a lower bound on the quantity $f(G_K)/f(O_K)$, known as the \textbf{performance bound}. 

\begin{remark}Our result on the performance bound will hold when the following two properties are true for our string matroid:

\begin{enumerate}
    \item Each component $o_i$ of the optimal strategy $O_K = (o_1, ..., o_K)$ exists as a string of length one in the string matroid.
    \item Each individual component $o_i$ of the optimal string is feasible for the greedy strategy computed up to stage $i-1$ for $i =1, ..., K$, in other words, $o_i$ is feasible for the string $G_{i-1} = (g_1, ..., g_{i-1})$. 
\end{enumerate}

\end{remark}
We can guarantee that these two conditions are satisfied for two classes of string matroids. The first is finite rank uniform string matroid, and the second is string matroid arising from a finite rank set matroid. We elaborate further on this in Sec. \rom{4}.

% \subsection{Translating From Matroids to String Matroids}
In the following two subsections, we introduce the notion of a finite rank set matroid to motivate the definition of a finite rank string matroid, and then show that a subset of rank $K$ string matroid can be obtained from a rank $K$ set matroid in a structure preserving way. 

% and then show that every rank $K < \infty$ set matroid can be identified with a subset of a
% string matroid of rank $K$ in a structure preserving way.

\subsection{Finite Rank String Matroid}
\begin{definition}
Let $N$ be any ground set, and $X$ a family of subsets of $N$. We say that
\textit{$(N, X)$ is a  finite rank $K$ \textbf{set matroid}} if 
\begin{enumerate}
    \item $|S| \leq K$, for all $S \in X$.
    \item $S \in X$ and $T \subset S$ implies $T \in X$ for all $S \in X$.
    \item For every $S, T \in X$ where $|T| +1 = |S|$, there exists $j \in S \setminus T$ such that $T \cup \{j\} \in X$.
\end{enumerate}
\end{definition}

\begin{remark}
The sets $S \in X$ are referred to as \textbf{independent sets} and an independent set of maximal length is known as a \textbf{basis}. The size of a basis is referred to as the \textbf{rank} of the matroid. In the present paper we will define set submodularity of functions on set matroids and their discrete derivatives in the same manner as we did for string matroids, with the exception that the partial ordering is changed from $\preccurlyeq$ to $\subset$, and the concatenation of strings is replaced with the union of sets. 
\end{remark}

Taking inspiration from the previous definition and from \cite{zhang2015string} we can define a \textit{finite rank string matroid}:

\begin{definition}
Let $\mathcal{A}$ be our action set, and $X \subset \mathcal{A}^*$. Then $X$ is a \textbf{finite rank string matroid} if
\end{definition}

\begin{enumerate}
    \item $|A| \leq K$ for all $A \in X$.
    \item If $B \in X$ and $A \preccurlyeq B$, then $A \in X$.
    \item For every $A, B \in X$ such that $|A|+1 = |B|$, there exists an $a \in \mathcal{A}^*$ which is a component of $B$ such that $A \cdot a \in X$.
\end{enumerate}

\begin{remark}
Similar to the above, we define the rank of $X$ to be the length of the largest strings which by the second axiom must all be the same size. When $X = \mathcal{A}^*_K$, then we say that $X$ is the \textbf{uniform string matroid of rank $K$}.
\end{remark}

\subsection{String Matroid Arising from Set Matroid}
We now discuss how we can identify a set matroid with a string matroid in a structure preserving way. Let $(X, N)$ be a set matroid of finite rank, and $S$ a set in $X$. Define the \textbf{string permutation set} of $S$ to be  $StrP(S)  = \{(s_1, ..., s_{|S|}) \in N^{|S|}: s_i \neq s_j \text{ when } i \neq j\ ,s_i \in S\}$. The universal action set whose actions come from the ground set $N$ of the finite rank set matroid $(N, X)$ will be denoted by $\mathcal{A}^*(N, X)$. We then define the map $\phi: X \to \mathcal{P}(\mathcal{A}^*(N, X))$ by $\phi(S) = StrP(S)$. The string matroid of rank $K$ corresponding to $(N, X)$ is then defined to be \[ \mathcal{A}_K^*(N, X) = \bigcup_{S \in X}\phi(S). \]

To establish some results, we will need to translate between $\mathcal{A}^*_K(N, X)$ and $(N, X)$ and to do so we define the map $\psi: \mathcal{A}^*_K(N, X) \to (N, X)$ by $\psi(A) = S$ where $A \in \phi(S)$ for some $S \in X$.
\begin{remark}
Some obvious properties of the map $\psi$ worth noting are
\begin{enumerate}
    \item If $A \preccurlyeq B \in \mathcal{A}_k^*(N, X)$, then $\psi(A) \subset \psi(B) \in X$. 
    \item For $S \in X$ and $\{j\} \in X$ such that $S \cup \{j\} \in X$, and $A \in \phi(S)$, $a = \phi(\{j\})$, we have $\psi(A\cdot a) = S\cup \{j\}$.
\end{enumerate}
\end{remark}

Given that we want to maximize functions on strings, we now define a way of translating functions on set matroids to functions on string matroids:

\begin{definition}
Let $f$ be a function on a matroid $(N, X)$. Then the \textbf{string extension} $\Tilde{f}: \mathcal{A}_K^*(N, X) \to \mathbb{R}$ is defined as $\Tilde{f}(A) := f(\psi(A))$.
\end{definition}

\begin{remark}
    Notice that in the way we defined our string extension, $\Tilde{f}$ is only determined by the components of the string and not their order. Therefore $\Tilde{f}$ does not take different values on strings which represent permutations of the same set.
\end{remark}

We may now prove our translation results.

\begin{lemma}
If $(N, X)$ is a finite rank $K$ set matroid, then $\mathcal{A}_K^*(N, X)$ is a finite rank $K$ string matroid.
\end{lemma}
\begin{proof}
The fact that  $\mathcal{A}_k^*(N, X)$ is rank $K$ follows from the fact that $(N, X)$ is rank $K$ and that any $A \in \phi(S)$ for $S \in X$ must satisfy $|A| \leq K$. Let $B \in \mathcal{A}^*_K(N, X)$, and $A \preccurlyeq B$. Then by the first observation in \textit{Remark} 5 we see that $\psi(A) \subset \psi(B)$, and since $\psi(B) \in X$, $\phi(\psi(B)) \subset \mathcal{A}^*_K(N, X)$, which means that $B \in \mathcal{A}^*_K(N, X)$ as desired. Lastly we see that if $A, B \in \mathcal{A}^*_K(N, X)$ are strings such that $|B| = |A| + 1$, then $|\psi(B)| = |\psi(A)| + 1$. By the third set matroid axiom we have the existence of an element of $\psi(B)$ which can be added to $\psi(A)$. This exact element can be concatenated to $A$ in order to produce another string in the string matroid. 
\end{proof}

\begin{lemma}
If $f$ is a set submodular function on a set matroid $(N,X)$, then the string extension $\Tilde{f}$ is string submodular on $\mathcal{A}_K^*(N, X)$. 
\end{lemma}
\begin{proof}
Note that the forward monotone condition is a result of both observations of \textit{Remark} 5 combined with the set submodularity of $f$.  Now suppose that $A \preccurlyeq B$, and $a$ is feasible for both $A$ and $B$, then 
$$
\begin{aligned}
\Tilde{f}(A \cdot a) - \Tilde{f}(A) = & f(\psi(A \cdot a)) - f(\psi(A)) \\
  \geq & f(\psi(B \cdot a)) - f(\psi(B))\\
  =  & \Tilde{f}(B \cdot a) - \Tilde{f}(B)\\
\end{aligned} 
$$
as desired. 
\end{proof}

\subsection{Summary and Assumptions for Following Results}
The main distinctions between a uniform string matroid and the one arising from a set matroid are the notions of permutation invariance of the actions and whether or not repetition of actions is allowed. In a uniform string matroid, different orders of the same actions can produce different outputs. Also, the same action can appear multiple times in a string. For a string matroid arising from a set matroid, the repetition of any action in a string is not allowed and the outputs of different orders the same actions are invariant to permutations.  

% In the set matroid case, the output of our objective function is solely determined by the actions in the set and not their order, whereas in the string matroid case a different order of the same actions can produce a different output. Secondly, in the string matroid case we have the possibility of a string which contains the same action multiple times, whereas this is not possible for a set matroid.   

The assumptions we have made about our functions and their domains so far have some recurring implications worth discussion here. All the analysis from the main results in Sec. \rom{4} is based on these assumptions. 

\begin{ass} $f$ is a forward monotone. 
\end{ass}
Being forward monotone guarantees that $\rho_a(A) \geq 0$. Combining this fact with the third string matroid property above, we obtain that the greedy algorithm will produce a string of the same length as the optimal string.

\begin{ass} 
 $f$ is submodular.
\end{ass}
 The second string matroid property guarantees that $A \in X$ for any $A \preccurlyeq B$ with $B \in X$. So we can always compute the discrete derivative along any singleton string. Combining this with the string submodularity assumption, we see that for any $a$ that is feasible for  a string $A$, $\rho_a(A) \leq \rho_a(\emptyset)$. In other words, the discrete derivative on the empty set bounds the discrete derivative on any larger string $A$ for which $a$ is feasible. The aforementioned fact also enables us to deduce that $\rho_a(A) / \rho_a(\emptyset) \leq 1$ when $\rho_a(\emptyset) > 0$. Lastly, notice that both of these properties guarantee that $\rho_{g_1}(\emptyset)$ is the largest possible discrete derivative of $f$.

\section{Previous Work}
 Most of the previous work regarding guarantees of the performance bound of the greedy strategy has been focused on the case where the domain is a finite set matroid, and can be placed into one of the following three categories. The work with string functions and string matroids happens to suffer from the same setbacks as the work on set matroids in category 2, so those results are included there.

\begin{cat}
\textbf{Classical Results and Algorithms} 
\end{cat}
Such results as in \cite{nemhauser1978analysis} guarantee that the greedy strategy will have an output of $f$ that is at least
$1/2$ of the output of $f$ on the optimal set for all set submodular functions on set matroids. When the set matroid is assumed to be uniform, we then have $1-e^{-1}$ for the performance bound. The downside of such results is that the greedy algorithm may yield high performance bound under weak subdmobularity. The above results only produce conservative bound in that case. 
% The downside of such results is that the lower bound on the performance can only be as good as the example in which the greedy algorithm performs the worst, and does not take into account how much better the greedy algorithm may perform in specific examples. 

\begin{cat} 
\textbf{Computationally Difficult Curvature}
\end{cat}
Other attempts at providing solutions to the problem, such as  \cite{wang2016approximation} and \cite{liu2019improved}, propose different notions of curvature which seek to measure the degree to which returns are diminished for the function $f$, and then establish bounds using these curvatures. Wang \textit{et al.} proposed the \textit{elemental curvature} in \cite{wang2016approximation}, which is defined as 
$$\alpha_{e
} = \max_{S\subset X, i, j \in X\setminus S, i \neq j}\frac{\rho_i(S \cup \{j\})}{\rho_i(S)},
$$ 
while in \cite{liu2019improved}, Liu \textit{et al.} introduced  notion of \textit{partial curvature} defined as 
$$
\alpha_{p} = \max_{j \in S \in X} \frac{\rho_j({S}\setminus \{j\})}{\rho_j(\emptyset)}.
$$ 
These curvatures are computable for very small matorids, but become computationally infeasible for large finite matroids. In a similar vein, the curvatures mentioned in \cite{zhang2015string} and \cite{streeter2008online} for string functions are also computationally intractable in general.

\begin{cat}
\textbf{Computationally Feasible Curvature}
\end{cat}
Conforti and Cornuéjols proposed the \textit{greedy curvature} in \cite{conforti1984submodular}, which is specific to each $f$ and its corresponding matroid $(N,X)$. This is the first computable curvature whose value is then used to yield the performance bound. 
Most recently, Welikala \textit{et al.} proposed \textit{extended greedy curvature}, which is also computable alongside the execution of greedy algorithm and specific to submodular functions and their set matroid domains \cite{welikala2022new}. Unfortunately, it requires computing greedy actions beyond the horizon $K$, which exceeds the domain on which the original problem is defined. \\
% Unfortunately, the algorithm presupposes that computations can be done for sets larger than the rank of the matroid, and so this  lacks the applicability of our bound. 
% In \cite{conforti1984submodular}, there is a computable lower bound which is specific to each $f$ and matroid $(N, X)$, which in many cases can be well above the $\frac{1}{2}$ guarantee of \cite{nemhauser1978analysis} giving a more accurate approximation to the performance of the greedy algorithm. Unfortunately, such a bound may fall below that of \cite{nemhauser1978analysis} in certain cases giving a less accurate picture of the performance.  \\
~\\
\indent The present work focuses on improving the greedy curvature bound in \cite{conforti1984submodular}. 
% In particular, Conforti and Cornuéjols provided a computable and problem specific lower bound for the performance of greedy algorithm which involves computing what they referred to as the \textit{Greedy curvature} \cite{conforti1984submodular}.
Specifically, the greedy curvature was defined as: 

$$
\alpha_G = 1- \min_{j \in N^i}\min_{1 \leq i \leq K}{ \frac{\rho_j(G_{i-1})}{\rho_j(\emptyset)}}
$$

\noindent where $N^i = \{j \in N \setminus G_{i-1}: G_{i-1} \cup \{j\} \in X, \rho_j(\emptyset) >0\}$. Allowing $\alpha$ to denote the second term in the expression for the greedy curvature above, we can write $\alpha_G = 1-\alpha$. 

Conforti and Cornuéjols then used the greedy curvature to prove the following lower bound on the performance of the greedy algorithm.

\begin{theorem}
The ratio $f(G_K) / f(O_K)$ is bounded below by  $1-\alpha_G \left( K-1 \right) / K$ \cite{conforti1984submodular}.
\end{theorem}

With some simple algebraic manipulation, we can write their lower bound as $1/K + \alpha \left( K-1 \right)/K$. Such a form for the bound will simplify the proofs presented in the next section.

\section{Main Results}

With most of the preliminaries out of the way, we can now state the main idea underlying this paper. In order to compute a lower bound for the ratio $f(G_K)/f(O_K)$, we find a computable upper bound for $f(O_K)$, i.e. $f(O_K) \leq B$, so that $f(G_K)/B \leq f(G_K)/f(O_K)$. The smaller the difference between $B$ and $f(O_K)$, the better the lower bound will be. 

Lemmata 1 and 2 guarantee that the computations performed involving the objective function defined on set matroids will be equal to computations involving the string extension of that objective function on the corresponding string matroid. Therefore for the rest of the paper, we will work in the uniform string matroid case. The conditions mentioned in \textit{Remark} 2 are satisfied for both rank $K$ uniform string matroids, as well as string matroids arising from finite rank set matroids. In the case of rank $K$ uniform string matroids, those conditions hold due to the fact that any action is feasible at any stage, and we allow for strings with repeated actions.

% The first case follows from the fact that any action is feasible at any stage, and we allow for strings with repeated actions.

As for the string matroids arising from set matroids, we only need the following lemma to establish this fact. When the lemma is applied to the basis of the optimal set in the corresponding set matroid, it guarantees that $(g_1, ..., g_{i-1}, o_i)$ is always a string in the string matroid.

\begin{lemma} The elements of any basis in a finite rank matroid $\Omega_K = \{\omega_1, ..., \omega_K\} \in X$ can be ordered so that $\rho_{\omega_i}(G_{i-1}) \leq \rho_{g_i}(G_{i-1})$. Furthermore, if $\omega_i \in G_K \cap \Omega_K$, then $\omega_i \equiv g_i$ \cite{conforti1984submodular}.
\end{lemma}

\begin{remark}
The proof of this statement is identical to the original proof presented in \cite{conforti1984submodular}, since the assumption of finite rank and the third matroid property guarantee that all bases will be the same finite size. Thus all that remains is to proceed with the same proof by induction.
\end{remark}

For the remainder of the paper we will assume that our string matroid $X$ is of finite rank $K$ and either (1) the uniform string matroid or (2) a string matroid arising from a set matroid.

In order to construct a lower bound that beats the greedy curvature bound in \cite{conforti1984submodular}, we create two upper bounds for $f(O_K)$ and choose the best of the two by taking their minimum. To introduce these new bounds we need the definitions that follow, wherein we frequently use elements from the set $X^G_i = \{ a\in X: \rho_a(\emptyset) > 0, a \text{ is feasible at } G_{i-1}\}$. The bounds we derive here dominate the greedy curvature bound in \cite{conforti1984submodular} when applied to both types of string matroids mentioned above. 

% string matroids arising from set matroids, and also apply to rank $K$ uniform string matroids in which the functions defined on them may be sensitive to permutations of the actions comprising the strings. 

\begin{definition}
 For $i =1, ..., K$, define $$\alpha_i := \min_{a \in X^G_i} \frac{\rho_a(G_{i-1})}{\rho_a(\emptyset)}. $$ 
\end{definition}

\begin{remark}From the discussion under Assumption 2, we see that a nondecreasing $f$ produces $\rho_a(A) \geq 0$ for all $A \in X$ and that $\rho_a(G_{i-1})/\rho_a(\emptyset) \leq 1$. The above observation tells us that $0 \leq \alpha_i \leq 1$ for $i=1, ..., K$ with $\alpha_1 = 1$ by definition. Lastly, $\alpha$ defined along with greedy curvature can be written as $\alpha = \min_{1 \leq i \leq K}\{\alpha_i\}_{i=1}^K $. 
\end{remark}

\begin{definition}
$$
 S(G_K, \alpha) := \begin{cases}
    \sum_{i=1}^K\frac{1}{\alpha_i}\rho_{g_i}(G_{i-1}) & \alpha > 0\\
    \infty & \alpha = 0
    \end{cases}
$$
\end{definition} 

We now introduce our first upper bound.

\begin{lemma}
The sum $S(G_K, \alpha)$ is an upper bound for $f(O_K)$. 
\end{lemma}
\begin{proof}
When $\alpha = 0$ the result is trivial. We argue that the following chain of inequalities holds:
$$
\begin{aligned}
f(O_K) & \stackrel{(1)}{=} \sum_{i=1}^K\rho_{o_i}(O_{i-1}) \stackrel{(2)}{\leq } \sum_{i=1}^K\rho_{o_i}(\emptyset) \\ & \stackrel{(3)}{\leq } \sum_{i=1}^K\frac{1}{\alpha_i}\rho_{g_i}(G_{i-1}).
\end{aligned}
$$ 
Equality (1) follows from a telescoping argument, and inequality (2) follows from the submodularity discussion in Assumption 2. To see inequality (3), we show that the second sum is termwise larger than the first. Being termwise greater is a result of the following three facts: (1) $\alpha_i\rho_{o_i}(\emptyset) \leq \rho_{o_i}(G_{i-1}) \leq \rho_{g_i}(G_{i-1})$  given that by the definition of $\alpha_i \leq \rho_{o_i}(G_{i-1}) / \rho_{o_i}(\emptyset)$; (2) $g_i$ maximizes the discrete derivative at $G_{i-1}$; and (3) $\rho_{o_i}(G_{i-1})$ exists when $X$ is a uniform string matroid or a string matroid arising from a set matroid. 
\end{proof}

We then construct the second upper bound for $f(O_K)$ that is useful in the event that $\alpha = 0$,  and exploits the fact that under our assumptions any action is feasible for the empty set.

\begin{definition}
Let $r_1 := g_1$, and
$$
r_i := \argmax_{a \in X\setminus \{r_1, ..., r_{i-1}\}} \rho_a(\emptyset) 
$$
for $i=2, ..., K$. Then define $R := \sum_{i=1}^K\rho_{r_i}(\emptyset)$.  
\end{definition}

We then see that $f(O_K) = \sum_{i=1}^K\rho_{o_i}(O_{i-1})  \leq \sum_{i=1}^{K}\rho_{o_i}(\emptyset) \leq \sum_{i=1}^K\rho_{r_i}(\emptyset) = R$ given that the largest $K$ values for $\rho_{r_i}(\emptyset)$ will bound all other discrete derivatives by the discussion in \textit{Assumption} 2, including those along the optimal set $O_K$.

In the next definition, we combine both of these upper bounds for $f(O_K)$ to obtain the best possible upper bound of the two.

\begin{definition}
 $$B := \sum_{i=1}^K\frac{1}{\beta_i}\rho_{g_i}(G_{i-1})$$ where $\beta_i = \begin{cases} 
            \alpha_i, & R \geq S(G_K, \alpha)\\
            \frac{\rho_{g_i}(G_{i-1})}{f(r_i)}, & R < S(G_K, \alpha) 
       \end{cases}$
       \; $i = 1,\cdots,K .$
\end{definition}

\begin{remark}
In a similar manner to the set of $\alpha_i$, we let $\beta = \min_{1 \leq i \leq K}\{\beta_i\}$. An important observation to make is that $f(O_K) \leq B \leq K\rho_{g_1}(\emptyset)$. The second inequality comes from the fact that all discrete derivatives are bounded by $\rho_{g_1}(\emptyset)$, so $K\rho_{g_1}(\emptyset)$ serves as a crude upper bound for $B$. 
\end{remark}

We now begin to establish that $f(G_K) / B$ is superior to the greedy curvature bound in \cite{conforti1984submodular}. 

\begin{theorem}
    The ratio $f(G_K) / f(O_K)$ is bounded below by $ f(G_K) / B$ with 
    $$ 
    \frac{f(G_K)}{B} \geq \frac{1}{K}  + \alpha \frac{ K-1}{K} = 1 - \alpha_{G} \frac{K-1}{K}. 
    $$
\end{theorem}

\begin{proof}
    As we defined above, $B$ serves the upper bound of $f(O_K)$. \\ 
    
    \noindent (a) First, we consider the case where $\alpha = 0$. \\
    The greedy curvature bound becomes $1 / K$ when $\alpha = 0$. Note that $f(G_K) / B \geq f(G_K) / \left( K\rho_{g_1}(\emptyset) \right)\geq 1 / K $, in which the last inequality is based on the fact that $f(G_K) \geq \rho_{g_1}(\emptyset) > 0$. \\

    \noindent (b) We then consider the case where $\alpha > 0$. \\
    By the definition of $\alpha$, we see $\alpha_{i} > 0 $ along with $\rho_{g_{i}}(G_{i-1}) > 0$ for all $1 \leq i \leq K$. We can rewrite the performance bound $f(G_K) / B $ in terms of $\beta$ to compare it with the greedy curvature bound in \cite{conforti1984submodular}:  
    $$
    \begin{aligned}
    \frac{f(G_K)}{B} & = \beta + \frac{f(G_K)}{B} - \beta \\ 
    & = \beta + \frac{1}{B} \left( \sum_{i=1}^K(1-\frac{\beta}{\beta_i})\rho_{g_i}(G_{i-1}) \right) \\
    & \stackrel {(4)}{\geq} \alpha + \frac{1}{B} \left( \sum_{i=1}^K(1-\frac{\alpha}{\alpha_i})\rho_{g_i}(G_{i-1}) \right) \\
    & \stackrel {(5)}{\geq} \alpha + \frac{1}{B}(1-\alpha)\rho_{g_1}(\emptyset) \stackrel {(6)}{\geq} \alpha + (1-\alpha)\frac{\rho_{g_1}(\emptyset)}{K\rho_{g_1}(\emptyset)}\\
    & = \frac{1}{K} + \alpha \frac{K-1}{K} = 1 - \alpha_{G}\frac{K-1}{K}. 
    \end{aligned}
    $$
    Inequality (4) follows from the fact that $\beta_{i}$ either equals $\alpha_{i}$ or $\rho_{g_i}(G_{i-1})/f(r_{i})$ for all $ i = 1,\cdots,K$, whichever produces smaller $B$, and $f(G_K) / B$ must be greater than the bound when $\beta_{i} = \alpha_{i}$. Inequality (5) holds since we only extract the first term and leave the rest nonnegative terms. Lastly, $B \leq K\rho_{g_1}(\emptyset)$ leads to the result from inequality (6). 
    
\end{proof}

\section{Application}
The multi-agent sensor coverage problem was originally proposed in \cite{zhong2011distributed} and further analyzed in \cite{sun2019exploiting} and \cite{welikala2022new}. In a given mission space, we aim to find an optimal placement of a set of sensors to maximize the chance of detecting randomly occurring events. Such a problem can be formulated as a set optimization problem and approximated by greedy solutions. In this section we illustrate the power of our results by applying them to a discrete version of the coverage problem. Our simplified version can be easily generalized to more complicated settings while remaining under the framework of set optimization. 

The mission space $\Omega \in \mathbb{R}^{2}$ is modeled as a non-self-intersecting polygon where $K$ homogeneous sensors should be placed to detect a randomly occurring event in $\Omega$. For simplicity of calculation, we assume both the sensors and the random event can only be placed and occur at lattice points. We denote the feasible space for sensor placement and event occurrence as $\Omega^{D}$. Our goal is to maximize the overall likelihood of successful detection in the mission space.

The likelihood of event occurrence over $\Omega^{D}$ is written as an event mass function $R: \Omega^{D} \xrightarrow{} \mathbb{R}_{\geq 0}$, and we assume that $\sum_{x \in \Omega^{D}} R(x) < \infty$. The outputs $R(x)$ may follow a particular distribution if some prior information is available. Otherwise $R(x) = 1$ when no prior information is obtained. The locations of all the sensors are represented as $\mathbf{s} = (s_1,s_2,\cdots,s_K) \in \prod_{i=1}^K\mathbb{Z}^2$. Each sensor has a finite sensing range with radius $\delta$ and is able to detect any occurring event with certain probability within the sensing range. Henceforth, the visibility region of a sensor located at $s_{i}$ is denoted by $V(s_i) = \{ x \; | \; \|x-s_i\| \leq \delta, \; x \in \Omega^{D}\}$. Additionally, the function $p(x,s_i) = e^{-\lambda \|x-s_i\|} \cdot \mathbf{1}_{\{x \in V(s_i)\}}(x)$ represents the probability of detecting an event occurring at $x$ for a sensor placed at $s_i$, where $\lambda$ is the decay rate characterizing how quick the sensing capability decays along the distance. 

Assuming all the sensors are working independently, the probability of detecting an occurring event at location $x \in \Omega^{D}$ after placing $K$ homogeneous sensors is $p(x,\mathbf{s}) = 1-\prod_{i=1}^{K}\left( 1-p(x,s_i) \right)$. Considering the whole feasible space, we need to employ the event mass function and our objective function of multi-agent sensor coverage becomes $H(\mathbf{s}) = \sum_{x \in \Omega^{D}} R(x)p(x,\mathbf{s})$. We aim to find $\mathbf{s} \in \prod_{i=1}^K\mathbb{Z}^2$ that maximizes $H(\mathbf{s})$:
\begin{equation*}
\label{obj_fun_sensor}
    \mathbf{s}^{*} = \arg\max_{\mathbf{s} \in \Gamma} H(\mathbf{s}), \text{ where }  \Gamma = \{S \subseteq \Omega^{D} : |S| \leq K \}.  
\end{equation*}
If $n$ lattice points in $\Omega^{D}$ are available for sensor placement, we therefore need to choose $K$ out of $n$ locations with complexity being $n!/\left( K!(n-K)! \right)$. This becomes a set optimization problem and exhaustive search is computationally intractable when $n$ is large. Therefore, greedy algorithm is an alternative approach for an approximation in polynomial time. 
It was proved that the continuous version of $H(\mathbf{s})$ is submodular in \cite{sun2019exploiting}, and it is not difficult to verify that its discrete version is also submodular. 
% Before we conduct our numerical experiments for greedy strategy, we demonstrate that the objective function for our sensor coverage problem (\ref{obj_fun_sensor}) is submodular. 
% \begin{theorem}
%     $H(\mathbf{s})$ is submodular. 
% \end{theorem}
% The theorem is proven for the continuous case in \cite{sun2019exploiting}, and hence must also hold for a discrete domain as in our problem. 

In our experiment, we consider a rectangular mission space of size $60 \times 50$ and $K = 10$ homogeneous sensors are required to be deployed. For a point $p = (x,y)$, the event mass function  is given by $R(x) = \left(x+y \right) / \left( x_{\text{max}}+y_{\text{max}} \right)$, where $x_{\text{max}} = 60$ and $y_{\text{max}} = 50$ are the largest values of the $x$ and $y$ component respectively for the mission space. Such a distribution implies that the randomly occurring event is more likely to happen in the top right corner of the rectangular mission space.  

%\begin{figure}[h!]
    %\centering
    %\includegraphics[width=0.52\textwidth]{sensor.png}
    %\caption{Demonstration of mission space: Blue dots are lattice points that are feasible for sensor placement and possible for event occurrence.} 
    %\label{sensors}
%\end{figure}

The comparison of different performance bounds is shown in Fig. \ref{sensor_bound}. A small decay rate implies good sensing capability and strong submodularity, under which the greedy strategy produces a low performance bound. Notice in Fig. \ref{sensor_bound} that our bound (red graph) always exceeds the greedy curvature bound (blue graph) as the theorem states. In addition, we can observe instances in which our bound is larger than the $1-e^{-1}$, while the greedy curvature bound is below this value. 

\begin{figure}[!ht]
    \centering
    \includegraphics[width=8.6cm]{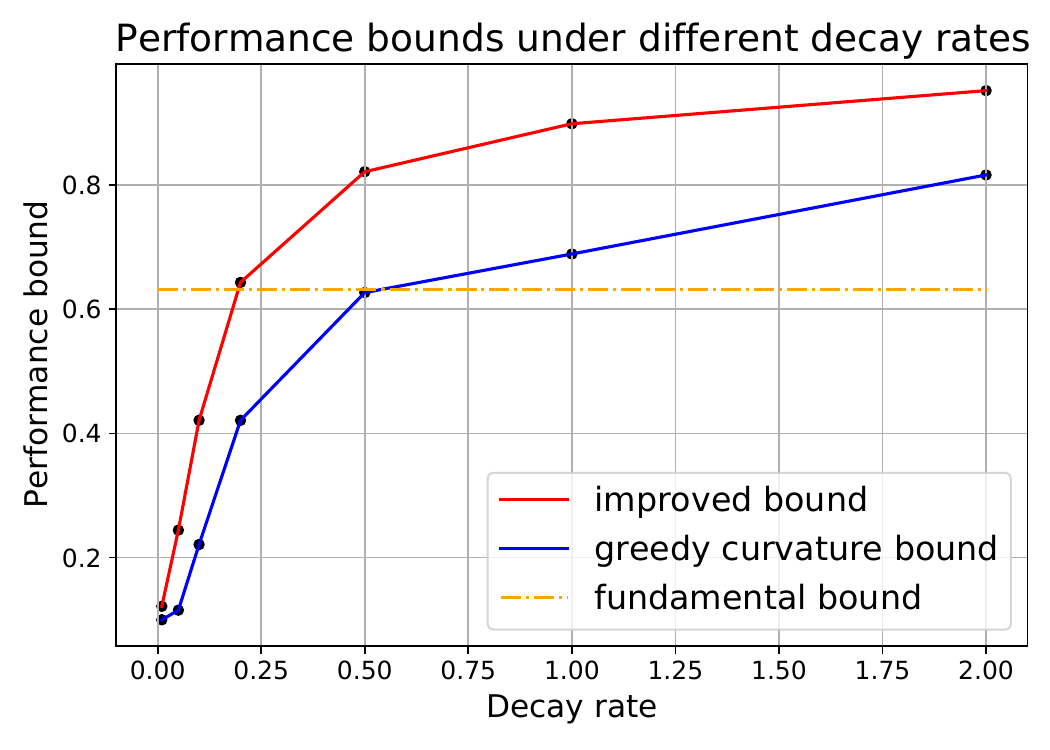}
    \caption{Comparison of performance bounds under different decay rates} 
    \label{sensor_bound}
\end{figure}

%%%%%%%%%%%%%%%%%%%%%%%%%%%%%%%%%%%%%%%%%%%%%%%%%%%%%%%%%%%%%%%%%%%%%%%%%%%%%%%%
\section{Conclusions and Future Work}
We introduced a new lower bound for the performance of the greedy strategy for string submodular functions on uniform string matroids of finite rank, and string matroids of finite rank arising from set submodular functions on set matroids. We then provided a proof that our lower bound is superior to the greedy curvature bound proposed by Conforti and Cornuejols in \cite{conforti1984submodular}. For future directions, we look to extend the bound to string matroids arising from independence systems, as well as applications of the bounds to reinforcement-learning problems. 

%%%%%%%%%%%%%%%%%%%%%%%%%%%%%%%%%%%%%%%%%%%%%%%%%%%%%%%%%%%%%%%%%%%%%%%%%%%%%%%%
\bibliographystyle{ieeetr}
\bibliography{root}

% \begin{thebibliography}{99}

% \bibitem{c1}
% J.G.F. Francis, The QR Transformation I, {\it Comput. J.}, vol. 4, 1961, pp 265-271.

% \bibitem{c2}
% H. Kwakernaak and R. Sivan, {\it Modern Signals and Systems}, Prentice Hall, Englewood Cliffs, NJ; 1991.

% \bibitem{c3}
% D. Boley and R. Maier, "A Parallel QR Algorithm for the Non-Symmetric Eigenvalue Algorithm", {\it in Third SIAM Conference on Applied Linear Algebra}, Madison, WI, 1988, pp. A20.

% \end{thebibliography}

\end{document}